\documentclass[conference]{IEEEtran}
\usepackage[margin=.72in,dvips]{geometry}
\makeatletter
\def\ps@headings{%
\def\@oddhead{\mbox{}\scriptsize\rightmark \hfil \thepage}%
\def\@evenhead{\scriptsize\thepage \hfil \leftmark\mbox{}}%
\def\@oddfoot{}%
\def\@evenfoot{}}
\makeatother

\usepackage{amsmath}
\usepackage{amssymb}
\usepackage{bbm}
\usepackage{graphics}
\usepackage{psfrag}
\usepackage{epsfig}
\usepackage[verbose]{cite}
\usepackage{algorithm}
\usepackage{algorithmic}
\usepackage{cite}
\newcommand{\ora}{\overrightarrow}

\newcommand{\E}{\ora{E}}

\newcommand{\ve}{\varepsilon}

\newcommand{\mce}{\mathcal{E}}

\newcommand{\mcm}{\mathcal{M}}
\newcommand{\mcr}{\mathbf{r}}
\newcommand{\mcp}{\mathbf{p_b}}

\newcommand{\mcc}{\mathcal{C}}

\newcommand{\ol}{\overline}

\newtheorem{lem}{Lemma}

\title{Study of Throughput and Latency in Finite-buffer Coded Networks}
\author{Nima Torabkhani$^\dag$,
     Badri N. Vellambi$^{\dag\ddag}$
     Faramarz Fekri$^\dag$\\
     $^\dag$ School of Electrical and Computer Engineering,
    Georgia Institute of Technology\\
     $^{\dag\ddag}$ Institute for Telecommunications Research,
    University of South Australia\\
     E-mail: \{nima, fekri\}@ece.gatech.edu, badri.vellambi@unisa.edu.au}

\begin{document}
\maketitle
\begin{abstract}
Exact queueing analysis of erasure networks with network coding in a finite buffer regime is an extremely hard problem due to the large number of states in the network. In such networks, packets are lost due to either link erasures or due to blocking due to full buffers. In this paper, a block-by-block random linear network coding scheme with feedback on the links is selected for reliability and more importantly guaranteed decoding of each block. We propose a novel method that iteratively estimates the performance parameters of the network and more importantly reduces the computational complexity compared to the exact analysis. The proposed framework yields an accurate estimate of the distribution of buffer occupancies at the intermediate nodes using which we obtain analytical expressions for network throughput and delay distribution of a block of packets.
\end{abstract}

\section{Introduction}\label{FB-Intro}

In networks, packets transfer through a series of intermediate nodes (routers). Hence, the packets may have to be stored at intermediate nodes for transmission at a later time, and often times, buffers are limited in size. Although a large buffer size is usually affordable and preferred to minimize packet drops, large buffers have an adverse effect on the latency of the network. Further, using larger buffer sizes at intermediate nodes will cause practical problems such as on-chip board space and increased memory-access latency.

The problem of buffer sizing and congestion control is of paramount interest to router design engineers. Typical routers today route several tens of gigabits of data each second \cite{FlowRef}. Realistic studies have shown that, at times, Internet routers handle about ten thousand independent streams/flows of data packets. With a reasonable buffer size of few Gigabytes of data, each stream can only be allocated a few tens of data packets. Therefore, at times when long parallel flows congest a router, the effects of such a small buffer space per flow come to play. This work aims at providing a framework to analyze the performance of single information flow in finite-buffer line networks and investigates the trade-offs between throughput, delay and buffer size.

The problem of computing throughput capacity and designing efficient coding schemes for lossy networks has been widely studied \cite{AmirDana:capacity, PakzadF05, YeungNCJrnl}. However, the study of capacity of networks with finite buffer sizes has been limited \cite{VelambiITW, VelambiITA, Torabkhani2010, NetCod:Lun}. This can be attributed solely to the fact that finite buffer systems are analytically hard to track. In \cite{VelambiITW, VelambiITA}, it was shown that min-cut capacity cannot be achieved due to the limited buffer constraint in a line network. They also proposed upper and lower bounds for the throughput capacity of wireline networks. Recently in \cite{Torabkhani2010}, authors presented an iterative estimation method to analyze the performance of a random routing scheme in general wired networks. Here, we will exploit similar ideas (on iterative estimation) to analyze the performance of block-based network coding scheme,\footnote{In the network coding literature, it is also referred as generation-based scheme, where each intermediate node combines only those packets that from the same generation.} which is a combination of both random linear coding and an ideal feedback. This is partly motivated by the feedback-based network coding scheme in\cite{ARQNC, FeedbackNC}. Further, the problem of finding the packet delay has been visited in queueing theory literature on the behavior of open tandem queues, which are analogous to line networks \cite{Tayfur_QT1, Tayfur_QT2}. However, the approaches therein cannot be used to study the performance of the random linear coding (RLC) scheme.

This work is motivated by the fact that the RLC scheme for finite-buffer networks introduced in \cite{NetCod:Lun} has the limitation that it cannot be used to characterize the latency profile. This is because, the typical notion of latency is not meaningful for the RLC scheme. Since latency is critical for real-time applications (such as video streaming), a block-by-block encoding of the stream is required. Hence, we introduce our \emph{Block-based Random Linear Coding} scheme, which applies RLC on each individual block\footnote{Just as in any network coding scheme, the packets received by the destination are linear combinations of the original data in a block. Hence, with the knowledge of this linear transformation at the decoder, inversion can be performed to recover the data block.}. We will see that our approach guarantees a certain decoding delay while this is not the case in RLC, which is a rate-optimal scheme with large decoding delay. Further, by using RLC, only one feedback is required for transmission of $K$ packets which considerably limits the average number of feedbacks per transmission.

This paper is organized as follows. First, we present a formal definition of the problem and the network model in Section~\ref{FB-sec1}. Next, we define the network states and the need for an approximation of the exact problem. We investigate the tools and steps for finite-buffer analysis in Section~\ref{FB-sec2}. We then obtain expressions for throughput and probability distribution of delay in Section~\ref{FB-sec3}. Finally, Section~\ref{FB-sec4} presents our analytical results compared to the simulations.

\section{Problem Statement and Network Model}\label{FB-sec1}

We consider line networks for our study. As illustrated in Fig.~\ref{FB_Fig1}, a line network is a directed graph of $h$ hops with the vertex set $V=\{v_0,v_1,\ldots,v_h\}$ and the edge set $\E=\{(v_i,v_{i+1}): i=0,\ldots,h-1\}$ for some integer $h\geq 2$. In the figure, the intermediate nodes are shown by black ovals. The links are assumed to be unidirectional, memoryless and lossy. We let $\ve_i$ denote the packet erasure probability over the link $(v_{i-1},v_i)$. The erasures model only the quality of links (e.g., presence of noise, interference) and do not represent packet drops due to finite buffers. Each node $v_i \in V$ has a buffer size of $m_i$ packets. It is assumed the destination node has no buffer constraints and that the source node has infinitely many innovative packets\footnote{A received packet is called \emph{innovative} with respect to a node if the packet cannot be generated by a linear combination of the current buffer contents of the node.}.
\begin{figure}[h]
\psfrag{e0}{$\ve_1$} \psfrag{e1}{$\ve_2$} \psfrag{e2}{$\ve_h$} \psfrag{m1}{\hspace{-3mm}$m_1$} \psfrag{m2}{\hspace{-3mm}$m_2$} \psfrag{m3}{\hspace{-5mm}$m_{h-1}$}
\psfrag{v0}{$v_0$} \psfrag{v1}{$v_1$} \psfrag{v2}{$v_2$} \psfrag{v3}{$v_{h-1}$} \psfrag{v4}{$v_h$}
\centering
\includegraphics[width=2.25in, height=1in,angle=0]{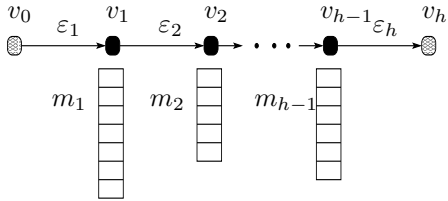}
\caption{An illustration of the line network.}
 \label{FB_Fig1}
\end{figure}
%---------------------------------------------------------------------------------------------------------------------
The system is analyzed using a discrete-time model, where each node transmits one packet over a link per epoch.
The unicast capacity between a pair of nodes is defined to be the supremum of all achievable rates of transmission of information packets (in packets per epoch) between the pair of nodes over all possible means used for packet generation and buffer update at intermediate nodes. The capacity is discussed thoroughly by the authors in \cite{Vellambi2010}. However, in this work we are not aiming to analyze the fundamental capacity of the network using rate optimal schemes. A practical network coding scheme, which is well-suited for transmitting real-time data streams in a block-by-block fashion using RLC and feedback, is introduced. The proposed scheme guarantees certain decoding delay, which is not the case in the RLC.

In the proposed practical scheme, the source node takes the stream of packets and divides them into blocks of $K$ packets each. The buffer of each intermediate node $v_i \in V$ is then segmented into $M_i$ blocks. In other words, we have $m_i=M_iK$. Each block is then served using RLC over all the packets in the block. The blocks are served based on a \emph{first-come first-serve} policy. An instant lossless hop-by-hop acknowledgment per block is also employed to indicate the successful receipt of a complete block of $K$ packets. In each epoch, one or multiple of the following events occur in different orders:
\begin{itemize}
\item[1.] If a node neither receives any innovative packet nor conveys any innovative packet to the next node, then the content of its buffer does not change\footnote{In this work, by \emph{conveying a packet}, we mean that \emph{the packet is successfully transmitted and stored at the next-hop node}.}.
\item[2.] Upon receiving an innovative packet, it will be stored in the last available block (a block of memory with less than $K$ innovative packets stored in it). The packets of this block will be served after all the previously received blocks in the buffer are completely served.
\item[3.] In each epoch, every node transmits a packet formed by the RLC encoding over its current block (oldest block in the queue), until the node receives an acknowledgment indicating that the block is fully conveyed to the next-hop node. The block will then be removed from the buffer and the next block in the queue will be served. This also implies that free space in the buffer will be increased by $K$ packets.
\end{itemize}

To implement the per-block feedback mechanism, a node must distinguish innovative packets upon their reception. One way is to compute the rank of the received packets in a block. A more practical protocol is to use a variable CMB in the header of each encoded packet, indicating the number of innovative packets used by RLC to form the packet. Further, every node maintains a counter INV indicating the number of innovative packets received in its current block so far. Every node sets INV = 0 for each new incoming block\footnote{If INV = $K$, then the counter is reset to 0 and an acknowledgment is sent to the previous node.} and increments it by one for each incoming packet whose CMB is greater than INV of the receiving node. Note that, if the current block in the queue of a node has $K$ innovative packets then CMB is equal to $K$ for all the packets to be transmitted by that node.

In this work, we will employ the following notations. For any $x\in[0,1]$, $\ol{x}\triangleq 1-x$. The convolution operator is denoted by $\otimes$ and $\otimes^{l} f$ is used as a shorthand for the $l$-fold convolution of $f$ with itself. For $0<\lambda<1$, $\mathbbm{G}(\lambda)$ denotes the probability mass function of a random variable that is geometric with mean $\frac{1}{1-\lambda}$.

\section{Exact analysis and Network States}\label{FB-core}

In \cite{VelambiITW}, a Markov-chain approach for exact analysis of a finite-buffer line network identifies the throughput as equivalent to the problem of finding the buffer occupancy distribution of the intermediate nodes. However, the size of the Exact Markov Chain (EMC) and the multiple reflections due to the finiteness of buffers at each intermediate node render this problem mathematically intractable for even networks of small hop-lengths and buffer sizes. We therefore aim to approximate the distribution of buffer occupancies.

To approach this approximation problem properly, it is necessary to clearly define the buffer states in a manner that (a) an irreducible ergodic Markov chain is obtained, and (b) the steady-state distribution of the chain allows tractable expressions for the performance parameters of the network. Thus, a proper definition for the buffer states cannot be proposed unless the communication scheme is known. For the scheme in Sec.~\ref{FB-sec1}, two variables are needed to track all the buffer states of a node. Let $s$ be the total number of innovative (w.r.t. the next-hop node) packets stored at a node. Denote $t$ to be the number of successfully conveyed innovative packets by the node from the current block. Then, the pair $(s, t)$ can be defined as the state of the buffer of the node. Note that $s$ is the minimum number of packets that a node has to store in its buffer. Also note that, since a new block starts to be served after the $K^{\textrm{th}}$ packet of the current block, $t\in\{0,\ldots,K-1\}$. As an example, assuming that node $v_i$ is in state $(s,t)$ at the start of the epoch, given that during the epoch it only sends a packet successfully but does not receive any packets, the state of the network will change to $(s,t+1)$ if $t=\{0, 1, \ldots, K-2\}$ or it will change to $(s-K,0)$ if $t=K-1$.

\section{A Framework for Finite-Buffer Analysis: Markov Chain Modeling}\label{FB-sec2}
In this section, we aim to determine the distribution of buffer occupancy of an intermediate node, which will later be used to analyze network parameters such as throughput and latency of a block.

Due to the discrete-time nature of the analysis framework, two Markov chains need to be constructed for each intermediate node. The first one considers the buffer occupancy at instants when a packet has just been transmitted (either successfully or unsuccessfully), which is called \emph{receive-first Markov chain} (RFMC). This is required to compute the probability of blocking, which is caused when the state of a node is forced to remain unchanged because the transmitted packet was successfully delivered to the next-hop node, but the latter does not store the packet due to full buffer occupancy . The second one considers the buffer occupancy at instants when a packet has just been received/stored, which is called \emph{transmit-first Markov chain} (TFMC). This will be used to calculate the incoming rate of innovative packets at each node.

Note that the problem of exactly identifying the steady-state probabilities of the RFMC and TFMC suffers the same difficulties as identifying that of the EMC~\cite{Vellambi2010}. The finite buffer condition introduces a strong dependency of state update at a node on the state of the node that is downstream. To develop an estimation scheme that considers blocking, we make the following assumptions.
\begin{itemize}
\item[1.] Packets are ejected from nodes in a memoryless fashion. This assumption allows us to keep track of only the information rate.
\item[2.] The blocking event occurs independent of the state of a node. This allows us to track just the blocking probability.
\item[3.] At any epoch, given the occupancy of a particular node, packet arrival and blocking events are independent of each other.
\end{itemize}

These assumptions spread the effect of blocking equally over all non-zero states of occupancy at each node. Now, given that the arrival rate of innovative packets at node $v_i$ is $r_i$ packets/epoch, and that the probability of the next node being full, ({i.e.}, $s=m_{i+1}$ for node $v_{i+1}$) is ${p_b}_{i+1}$, we can show that transition dynamics of the state change for node $v_i$ is given by the Markov chain depicted in Fig.~\ref{TFMC} for both TFMC and RFMC \footnote{Self-loops are not demonstrated in the figure.}.
\begin{figure}[ht!]
\centering
\includegraphics[height=3.5in,angle=-90]{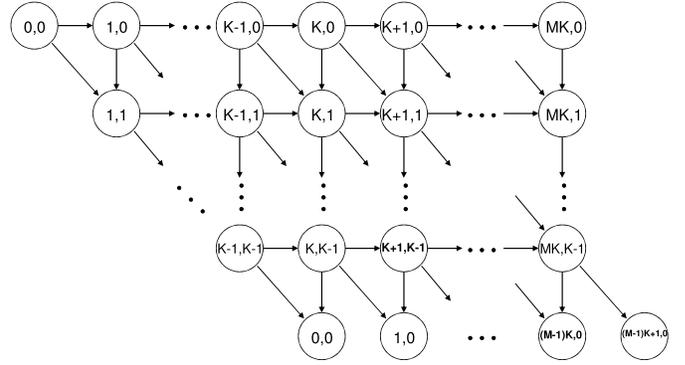}
\caption{The general structure for both TFMC and RFMC for a node with buffer size $m=MK$.}\label{TFMC}
\end{figure}
Also, obtaining the state transition probabilities is straightforward using $r_i$, ${p_b}_{i+1}$, $\ve_i$ and $\ve_{i+1}$. As an example, for TFMC we have the following
\begin{equation}
P^{TF}_{(s,t) \to (s+1,t)} =
\left\{\hspace{-1.5mm}
\begin{array}{ll}
r_i & s=t \\
r_i (\ve_{i+1} + \ol\ve_{i+1} {p_b}_{i+1} ) & s-1 \geq t\\
0 & \text{Otherwise}
\end{array}
\right..
\label{Formula_TF}
\end{equation}

Note that we notate $s$ and $t$ instead of $s_i$ and $t_i$ for the simplicity of notation when considering node $v_i$. The same transition probability for RFMC is different from TFMC and is given by 
\begin{equation}
P^{RF}_{(s,t) \to (s+1,t)} =
r_i (\ve_{i+1} + \ol\ve_{i+1} {p_b}_{i+1} ).
\label{Formula_RF}
\end{equation}

Note that, both (\ref{Formula_TF}) and (\ref{Formula_RF}) are valid for $s=\{0, 1, \ldots, M_iK-1\}$ and $t=\{0, 1, \ldots, K-1\}$.

For all input parameters, the Markov chains can be shown to be aperiodic, irreducible and ergodic. Therefore, it possesses a unique steady-state distribution. The steady state probability of node $v_i$ being in state $(s,t)$, is denoted by $P_{i}^{\emph{RF}}(s,t)$ and $P_{i}^{\emph{TF}}(s,t)$ for RFMC and TFMC, respectively.

The blocking probability that the node $v_{i-1}$ perceives from the node $v_i$ is the same as the probability of $v_i$ being full ($s=M_iK$) at the instant when a packet is transmitted successfully by the previous node. Hence, this probability have to be calculated using the steady state probability distribution of RFMC as follows
\begin{equation}
{p_b}_i=\sum_{t=0}^{K-1} P_{i}^{\emph{RF}}(M_iK,t).
\label{P_block}
\end{equation}

Similarly, the steady state probability distribution of TFMC can be used to compute the arrival rate at the next node using
\begin{equation}
r_{i+1}=(1-\sum_{t=0}^{K-1} P_{i}^{\emph{TF}}(t,t))\ol\ve_{i+1}.
\label{arrival_rate}
\end{equation}

Note that if a node is in the state $(t,t)$ ($t \in \{0, \ldots, K-1\}$), it means that it has stored $t$ innovative packets from the current block so far and it has also sent $t$ linear combinations of them successfully. Therefore, there is no more innovative packets to send.

Finally, Given two vectors $\mcr=(r_1,\ldots,r_h)\in[0,1]^{h}$ and $\mcp=({p_b}_1,\ldots,{p_b}_h)\in[0,1]^{h}$, we term $(\mcr,\mcp)$ as an approximate solution to the exact problem, if they satisfy (\ref{P_block}) and (\ref{arrival_rate}) in addition to having $r_1=\ol\ve_1$ and ${p_b}_h=0$. Fortunately, the result from our previous work on the capacity of line networks in~\cite{Vellambi2010} can be applied to this setup as well and guarantees both the uniqueness and an algorithm for identifying the approximate solution.

The approximate solution is obtained iteratively by the following procedure:

\begin{itemize}
\item[1.] Initialization: $\mcr=(\overline\ve_1,\ldots,\overline\ve_h)$ and $\mcp=(0,\ldots,0)$

\item[2.] Construct TFMC and RFMC respectively and compute their steady-state probabilities.

\item[3.] Using $P_{i}^{\emph{RF}}(s,t)$ and $P_{i}^{\emph{TF}}(s,t)$ (obtained from step 2) calculate the new values for $\mcr$ and $\mcp$ by (\ref{P_block}) and (\ref{arrival_rate}), respectively, for $i=1,2, \ldots, h-1$  and auxiliary equations $r_1=1-\ve_1$ and ${p_b}_h=0$.

\item[4.] Repeat steps 2 and 3 until all the distributions converge.
\end{itemize}

\section{Computation of Network Parameters}\label{FB-sec3}

In this section, we exploit the results of the iterative estimation of buffer occupancy distributions in Sec.~\ref{FB-sec2} to obtain analytical expressions for both network throughput and delay distribution of a block. The \emph{Block Delay} is defined as the time taken for a block of $K$ information packets (at the source) to be transferred through a line network from the instant when the first packet of that block is transmitted from the source node to the instant when the $K^{\textrm{th}}$ innovative packet of that block is received by the destination node ({i.e.}, the block can be decoded at the destination).

Given a line network with link erasures $\mce=(\ve_1,\ldots,\ve_h)$, intermediate node buffer sizes $\mcm=(m_1,\ldots, m_{h-1})$, we can find the approximate solution $(\mathbf{r},\mathbf{p_b})$. Using this, an estimate of the throughput is obtained using the following 
\begin{equation}
\mcc(\mce,\mcm,K)=r_h(1-{p_b}_h)=r_h.
\label{FB-LinCapEst}
\end{equation}

To compute the distribution of the total \emph{block delay}, one can proceed in a hop-by-hop fashion in the following way:
\begin{equation}
\mathbf{D}=\mathbf{T}_{1} \otimes \mathbf{W}_{1} \otimes \mathbf{W}_{2} \otimes \ldots \otimes \mathbf{W}_{h-2} \otimes \mathbf{F},
\label{Delay_dist_formula}
\end{equation}

 where $\mathbf{T}_{1}$ is the probability distribution of the time taken for a packet in the source to be conveyed to node $v_1$. Further, $\mathbf{W}_{i}$ is the probability distribution of the time taken from the instant when node $v_i$ stores the first innovative packet of a block to the instant the first packet of the corresponding block in $v_i$ is conveyed to node $v_{i+1}$. Finally $\mathbf{F}$ is the probability distribution of the time taken for all the $K$ packets of a block in node $v_{h-1}$ to be conveyed to the destination node from the instant when the first innovative packet of the same block is stored in the buffer of node $v_{h-1}$. Thus, using the definition, we have $\mathbf{T}_{1} = \mathbbm{G}(\acute{\ve_1})$, where $\acute{\ve_1}$ is the effective erasure probability (after considering blocking) and is given by
\begin{equation}
\acute{\ve_i}=
\left \{
\begin{array} {ll}
\ve_i + {p_b}_i \ol\ve_i & i=1,2,\ldots,h-1\\
\ve_h & i=h
\end{array}
\right..
\label{eps_eff}
\end{equation}

Also, the average waiting time in node $v_i$ is formulated as
\begin{equation}
\mathbf{W}_i=
\pi_i(0,0)\mathbf{S}_i(0,0) + \sum_{d=1}^{M_i-1} { \sum_{t=0}^{K-1} {\pi_i(dK,t)\mathbf{S}_i(dK,t)} },
\label{W_i}
\end{equation}
where, $\pi_i(s,t)$ is the probability that an arriving packet finds node $v_i$ in state $(s,t)$ given that it is the first packet of its corresponding block. Also, $\mathbf{S}_i(s,t)$ is the probability distribution of the time taken for the first innovative packet of a block in $v_i$ to be conveyed to node $v_{i+1}$ from the instant when the first innovative packet of that block arrives at node $v_i$ and finds its buffer at state $(s,t)$.

If an arriving packet is the first of its corresponding block, it finds the buffer at states of the form $(dK,t)$ where $d$ can take any value between $0$ and $M_i-1$. This is because of the fact that the last block had been completely served before the first packet of the current block arrives. Hence, both  $\mathbf{S}_i(s,t)$ and  $\pi_i(s,t)$ will be 0 if $s$ is not a multiple of $K$. Finally,  $\pi_i(s,t)$ can be formulated as follows for $n\in\{0,1,\ldots,M_i-1\}$,
\begin{equation}
\pi_i(s,t)=
\left \{
\begin{array} {ll}
\frac{P_{i}^{\emph{RF}}(s,t)}{\sum_{d=0}^{M_i-1} {\Phi^{RF}_i(dK)}} & s=nK\\
0 & \text{Otherwise}
\end{array}
\right.,
\label{pktsee}
\end{equation}
where $\Phi^{\emph{RF}}_{i}(s)$ is the marginal probability distribution of $s$ for an arriving packet (\emph{i.e.}, The probability that an arriving packet finds the buffer of node $v_i$ in the states of the form $(s,.)$). $\Phi_{i}$ can be computed as follows
\begin{equation}
\Phi_{i}(s)=\sum_{t=0}^{\min\{K-1,s\}} {P_{i}^{\emph{RF}}(s,t)}.
\label{phi_i}
\end{equation}

Also, $\mathbf{S}_i(s,t)$ can be derived using the following relation for $n\in\{0,1,\ldots,M_i-1\}$
\begin{equation}
\mathbf{S}_i(s,t)=
\left \{
\begin{array} {ll}
\mathbbm{G}(\acute{\ve_i}) & s=0 \\
\otimes^{\{(n-1)K + (K+1-t)\}}\mathbbm{G}(\acute{\ve_i}) & s=nK\\
0 & \text{Otherwise}
\end{array}
\right..
\label{Delay_Func_1}
\nonumber
\end{equation}

Finally, $\mathbf{F}$ can be calculated by taking the average over all the conditional delay distributions $\mathbf{L}(s,t)$ as
\begin{equation}
\mathbf{F}=
\pi_{h-1}(0,0)\mathbf{L}(0,0) + \sum_{d=1}^{m_{h-1}-1} { \sum_{t=0}^{K-1} {\pi_{h-1}(dK,t)\mathbf{L}(dK,t)} },
\label{F}
\nonumber
\end{equation}
where $\mathbf{L}(s,t)$ is the probability distribution of the time taken for the whole block to be conveyed to the destination node given that the first packet of that block found the buffer of node $v_{h-1}$ in state $(s,t)$ when arrived. Note that after receiving the first packet of a block, node $v_{h-1}$ has to wait until all its previously stored blocks are conveyed to the destination, during which some of the packets of the corresponding block might have already been arrived. Let $\mathbf{V}(x,y)$ be the probability distribution of the time to convey $y$ innovative packets to the next node when $x$ of those packets ($x \leq y$) has yet to arrive for the same block, knowing that a packet departs with probability $P_{out}$ and an innovative packet arrives with probability $P_{in}$ in each time epoch. Hence, $\mathbf{L}(s,t)$ is derived using the following relation for $n\in\{0,1,\ldots,M_i-1\}$
\begin{equation}
\mathbf{L}(s,t)=
\left \{
\begin{array} {ll}
\mathbf{V}(K-1,K) & s=0 \\
\{\otimes^{\alpha(n,t)}\mathbbm{G}(\ve_h) \}+\mathbf{V}(K-\beta,K) & s=nK\\
0 & \text{otherwise}
\end{array}
\right.,
\label{Delay_Func_2}
\nonumber
\end{equation}
where $\alpha(n,t)=(n-1)K + (K-t)$ is the number of packets that have to leave node $v_{h-1}$ before the next block to be served and $\beta=\min\{K-1,\left\lfloor \frac{\alpha(n,t) P_{in}}{P_{out}}\right\rfloor \}$ is the expected number of packets from the corresponding block that arrived during the time when those $\alpha(n,t)$ packets were being conveyed. Further, $P_{in}=\ol\ve_{h-1}(1-\sum_{t=0}^{K-1}P_{h-1}^{\emph{TF}}(t,t))$ for $h>2$, $P_{in}=\ol\ve_{h-1}$ for $h=2$, and $P_{out}=\ol\ve_h$. $\mathbf{V}(x,y)$ will be determined by the following lemma.

%%%%%%%%%%%%%%%%%%%%%%%%%%%%%%%%%%%%%%%%%%%%%%%%%%%%%%%%%%%%%%%%%%%%%%%%%%%%%%%%%%%%%%%
\vspace{2mm}
\begin{lem}\label{FB-lem1}
$\mathbf{V}(x,y)$ (defined for $x \leq y$) is the solution to the following equation,
{\scriptsize{
\begin{equation}
\mathbf{V}(x,y)=\big[\frac{p_1\mathbf{V}(x,y-1)+p_2\mathbf{V}(x-1,y)+p_3\mathbf{V}(x-1,y-1)}{p_1+p_2+p_3}\big]\otimes\mathbbm{G}(p_4)
\nonumber
\end{equation}}}
with boundary conditions:
\begin{equation}
\begin{array} {lcl}
\mathbf{V}(0,y) &=& \otimes^{y}\mathbbm{G}(1-P_{out})\\
\mathbf{V}(x,x) &=& \mathbbm{G}(1-P_{in}) \otimes \mathbf{V}(x-1,x)
\end{array},
\end{equation}

where,
\begin{equation}
\begin{array}{lcl}
p_1 = P_{out}(1-P_{in})&  &p_2 = P_{in}(1-P_{out})\\
p_3 = P_{in}P_{out} &  & p_4 = (1-P_{in})(1-P_{out}).
\end{array}
\end{equation}
\end{lem}

\begin{proof}
The proof is omitted due to space limitations.
\end{proof}

\section{Results of Simulation}\label{FB-sec4}

In this section, we compare our analytical results to the actual simulations. To study the effect of buffer size on throughput and block delay, we simulated a line network of eight hops for two cases where all the links have the same probability of erasure of $0.1$ or $0.2$. The buffer size $m$ (in packets) is divided into $M$ blocks of $K$ packets.
\begin{figure}[ht]
\centering
\small
%\hspace*{-1.0cm}
\psfrag{xaxis1}{\small{$m$ (packets)}}
\psfrag{yaxis1}{\hspace{-15pt}Throughput}
\psfrag{xaxis2}{\small{$m$ (packets)}}
\psfrag{yaxis2}{\hspace{-15pt}Average Delay}
\psfrag{dddddddddddddddddddata1}{\tiny{Estimation ($\ve=0.1$)}}
\psfrag{data2}{\tiny{Simulation ($\ve=0.1$)}}
\psfrag{data3}{\tiny{Estimation ($\ve=0.2$)}}
\psfrag{data4}{\tiny{Simulation ($\ve=0.2$)}}
\psfrag{dddddddddddddddddddata5}{\tiny{Estimation ($\ve=0.1$)}}
\psfrag{data6}{\tiny{Simulation ($\ve=0.1$)}}
\psfrag{data7}{\tiny{Estimation ($\ve=0.2$)}}
\psfrag{data8}{\tiny{Simulation ($\ve=0.2$)}}
\epsfig{height=6cm,width=8.5cm, file=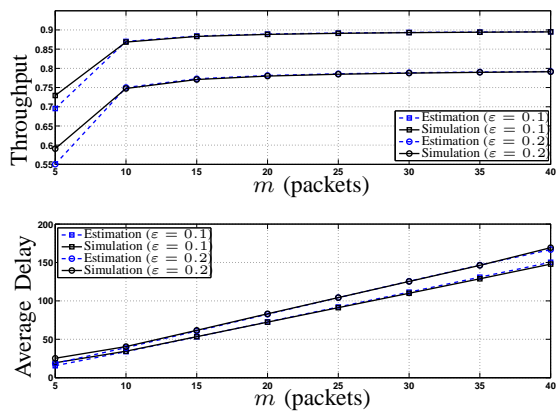}
\caption{Delay and throughput of an 8-hop line network as a function of $m$.}\label{FB-TPT-AD}
\end{figure}
%----------------------------------------------------
Fig.~\ref{FB-TPT-AD} presents the variation of our analytical results and the actual simulations for both throughput and average delay of a block, as the buffer size $m$ of the intermediate nodes is varied while the block size is fixed at $K=5$ packets. It can be seen that as the buffer size is increased, average delay also increases linearly. It appears that above memory sizes of $10$, the gain in capacity is negligible, while the latency increases significantly. Hence, there is no need to allocate more storage to the flow even if the space is available in the router. Further, for buffer sizes of less than $10$ packets, there is a gap from the min-cut capacity, a diverging point from asymptotic results due to the finite buffer effect.
\begin{figure}[ht]
%\hspace*{-1.0cm}
\psfrag{xaxis}{\small{\hspace{-20pt}Delay of a block (epochs)}}
\psfrag{yaxis}{\hspace{-15pt}Probability of Delay}
\psfrag{ddddddddddddddd1}{\tiny{Simulation ($m=16$)}}
\psfrag{d2}{\tiny{Estimation ($m=16$)}}
\psfrag{d3}{\tiny{Simulation ($m=24$)}}
\psfrag{d4}{\tiny{Estimation ($m=24$)}}
\psfrag{d5}{\tiny{Simulation ($m=32$)}}
\psfrag{d6}{\tiny{Estimation ($m=32$)}}
\epsfig{height=6cm,width=9.0cm,file=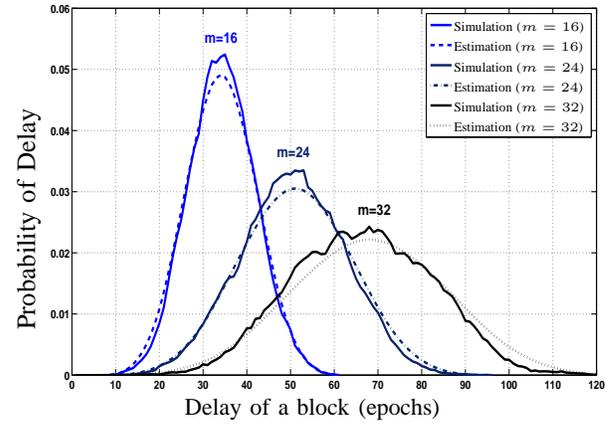}
\caption{Delay profiles of a 5-hop line network for varying buffer sizes ($K=4$).}\label{FB-DD}
%\vspace{.07in}
\end{figure}
%----------------------------------------------------
Fig.~\ref{FB-DD} presents a comparison between the actual and the estimated delay profile for a five-hop line network with the erasure probability on every link set to $0.05$. It also compares the delay profiles for different buffer sizes when $K = 4$. It is noticed that as the buffer size of the intermediate nodes is increased, both average delay and its standard deviation are increased. This is undesirable since any increase in the standard deviation of the delay can make congestion control algorithms unstable.

\bibliographystyle{ieeetr}
\bibliography{AsilomarRef}

\begin{thebibliography}{10}

\bibitem{FlowRef}
G.~Appenzeller, I.~Keslassy, and N.~McKeown, ``Sizing router buffers,'' {\em
  SIGCOMM Comput. Commun. Rev.}, vol.~34, no.~4, pp.~281--292, 2004.

\bibitem{AmirDana:capacity}
A.~F. Dana, R.~Gowaikar, R.~Palanki, B.~Hassibi, and M.~Effros, ``Capacity of
  wireless erasure networks,'' {\em IEEE Transactions on Information Theory},
  vol.~52, no.~3, pp.~789--804, 2006.

\bibitem{PakzadF05}
P.~Pakzad, C.~Fragouli, and A.~Shokrollahi, ``Coding schemes for line
  networks,'' {\em IEEE International Symposium on Information Theory (ISIT
  2005)}, Sep. 2005.

\bibitem{YeungNCJrnl}
S.-Y.~R. Li, R.~W. Yeung, and N.~Cai, ``Linear network coding,'' {\em IEEE
  Transactions on Inform. Theory}, vol.~49, pp.~371--381, February 2003.

\bibitem{VelambiITW}
B.~N. Vellambi, N.~Rahnavard, and F.~Fekri, ``The effect of finite memory on
  throughput of wireline packet networks,'' {\em in the Proc. of Information
  Theory Workshop (ITW 2007), Lake Tahoe, CA}, Sept.~2007.

\bibitem{VelambiITA}
B.~N. Vellambi and F.~Fekri, ``On the throughput of acyclic wired packet
  networks with finite buffers,'' {\em in the Proc. of 2008 Information Theory
  and Applications Workshop (ITA 2008), San Diego, CA}, Jan.~2008.

\bibitem{Torabkhani2010}
N.~Torabkhani, B.~N. Vellambi, and F.~Fekri, ``Throughput and latency of
  acyclic erasure networks with feedback in a finite buffer regime,'' {\em in
  the Proc. of Information Theory Workshop (ITW 2010), Dublin, Ireland}, Aug.
  2010.

\bibitem{NetCod:Lun}
D.~S. Lun, P.~Pakzad, C.~Fragouli, M.~M\'{e}dard, and R.~Koetter, ``An analysis
  of finite-memory random linear coding on packet streams,'' {\em in Proc. of
  NetCod 2006, Boston, MA, April 3-7, 2006}.

\bibitem{ARQNC}
J.~Kumar~Sundararajan, D.~Shah, and M.~Medard, ``{ARQ} for network coding,''
  {\em IEEE International Symposium on Information Theory (ISIT 2008)},
  pp.~1651 --1655, jul. 2008.

\bibitem{FeedbackNC}
C.~Fragouli, D.~Lun, M.~Medard, and P.~Pakzad, ``On feedback for network
  coding,'' {\em 41st Annual Conference on Information Sciences and Systems
  (CISS 2007)}, pp.~248 --252, mar. 2007.

\bibitem{Tayfur_QT1}
T.~Altiok, ``Approximate analysis of queues in series with phase-type service
  times and blocking,'' {\em Operations Research}, vol.~37, pp.~301--310, July
  1989.

\bibitem{Tayfur_QT2}
T.~Altiok, ``Approximate analysis of exponential tandem queues with blocking,''
  {\em European Journal of Operational Research}, vol.~11, no.~4, pp.~390--398,
  1982.

\bibitem{Vellambi2010}
B.~N. Vellambi, N.~Torabkhani, and F.~Fekri, ``Throughput and latency in
  finite-buffer line networks,'' {\em accepted in IEEE Transaction on
  Information Theory}, June~2010.

\end{thebibliography}
\end{document}